\documentclass{elsarticle}
\usepackage{lipsum}
\makeatletter
\def\ps@pprintTitle{%
 \let\@oddhead\@empty
 \let\@evenhead\@empty
 \def\@oddfoot{}%
 \let\@evenfoot\@oddfoot}
\makeatother
\usepackage[utf8]{inputenc}
\usepackage[english]{babel}
\usepackage{amsfonts, amssymb, amsmath, amsthm, gensymb, mathcomp, stmaryrd}
\usepackage[ruled, vlined, linesnumbered]{algorithm2e}
\usepackage[obeyDraft]{todonotes}
\usepackage{tikz}
\usepackage{pgfplots}

\newcommand{\floor}[1]{\left\lfloor #1 \right\rfloor}

\newtheorem{theorem}{Theorem}[section]       
\newtheorem{property}{Property}[section]     

\title{New modular multiplication and division algorithms based on continued
  fraction expansion}

\author[upmc]{Mourad Gouicem} 
\address[upmc]{UPMC Univ Paris 06 and CNRS UMR 7606, LIP6 \\4 place Jussieu,
  F-75252, Paris cedex 05, France
}

\begin{document}

\begin{abstract}
  In this paper, we apply results on number systems based on continued fraction
  expansions to modular arithmetic. We provide two new algorithms in order to
  compute modular multiplication and modular division. The presented algorithms
  are based on the Euclidean algorithm and are of quadratic complexity.
\end{abstract}


\maketitle

\section{Introduction}
Continued fractions are commonly used to provide best rational approximations of
an irrational number. This sequence of best rational approximations
$(p_i/q_i)_{i\in\mathbb{N}}$ is called the convergents' sequence. In the
beginning of the 20$^{th}$ century, Ostrowski introduced number systems derived
from the continued fraction expansion of any irrational $\alpha$
\cite{berthe2009diophantine}. He proved that the sequence
$(q_i)_{i\in\mathbb{N}}$ of the denominators of the convergents of any
irrational $\alpha$ forms a number scale, and any integer can be uniquely
written in this basis. In the same way, the sequence $(q_i\alpha -
p_i)_{i\in\mathbb{N}}$ also forms a number scale.

In this paper, we show how such number systems based on continued fraction
expansions can be used to perform modular arithmetic, and more particularly
modular multiplication and modular division. The presented algorithms are of
quadratic complexity like many of the existing implemented algorithms
\cite[Chap. 2.4]{brent2010modern}. Furthermore, they present the advantage of
being only based on the extended Euclidean algorithm, and to integrate the
reduction step.

In the following, we will first introduce notations and some properties of the
number systems based on continued fraction expansions in Section
\ref{sect_num_sys_cf}. Then we describe the new algorithms in Section
\ref{sect_mod_arith}. Finally, we give elements of complexity analysis of these
algorithms in Section \ref{sect_comp}, and perspectives in Section
\ref{sect_persp}.


\section{Number systems and continued fractions}
\label{sect_num_sys_cf}
\subsection{Notations}
First, we give some notations on the continued fraction expansion of an
irrational $\alpha$ with $0<\alpha<1$ \cite{Khinchin92}. We call the
\emph{tails} of the continued fraction expansion of $\alpha$ the real sequence
$(r_i)_{i\in\mathbb{N}}$ defined by
\begin{align*}
  r_0 & = \alpha,\\
  r_i & = 1/r_{i-1} - \lfloor 1/r_{i-1} \rfloor.
\end{align*}
We denote $(k_i)_{i\in\mathbb{N}}$ the integer sequence of the partial quotients
of the continued fraction expansion of $\alpha$. They are computed as $k_i =
\lfloor 1/r_{i-1} \rfloor$. We have
\[
\alpha = \cfrac{1} 
{k_1 + 
  \cfrac{1} 
  {k_2 + 
    \cfrac{1}
    {\ddots + 
      \cfrac{1}
      {k_{i} + {r_i}}}}} := [0; k_1, k_2, \dots, k_{i} +  r_i].
\]

We write ${p_i}/{q_i}$ the $i^{th}$ convergent of $\alpha$. The sequences
$(p_i)_{i\in\mathbb{N}}$ and $(q_i)_{i\in\mathbb{N}}$ are integer valued and positive,
\[
\frac{p_i}{q_i} = 
[0; k_1, k_2, \dots, k_i].
\]

We will also write $(\theta_i)_{i\in\mathbb{N}}$ the positive real sequence of
$(-1)^{i}(q_i\alpha - p_i)$ which we call the sequence of the \emph{partial
  remainders} as they are related to the tails by $r_i = \theta_i/\theta_{i-1}$.
Hereafter, we recall the recurrence relations to compute these sequences,
\[
\begin{array}{lll}
  p_{-1} = 1 & p_0 = 0 & p_i= p_{i-2} + k_ip_{i-1}, \\
  q_{-1} = 0 & q_0 = 1 & q_i= q_{i-2} + k_iq_{i-1}, \\
  \theta_{-1} = 1 & \theta_0 = \alpha & \theta_i = \theta_{i-2} - k_i\theta_{i-1}.
\end{array}
\]

We also write $\eta_i = q_i\alpha - p_i$ the sequence of the \emph{signed
  partial remainders}, which elements are of sign $(-1)^i$.  The sequence
$(\eta_i)_{i\in\mathbb{N}}$ of the signed partial remainders can be computed as
$((-1)^i\theta_i)_{i\in\mathbb{N}}$.


\subsection{Related number systems over irrational numbers}
In this section, we present two number systems based on the sequences of the
signed partial remainders $(\eta_i)_{i\in\mathbb{N}}$ and the denominators of
the convergents $(q_i)_{i\in \mathbb{N}}$ of an irrational $\alpha$. They have
been extensively studied during the second part of the $20^{th}$ century
\cite{berthe2009diophantine,Vershik1994}.

\begin{property}[{\cite[Proposition 1]{berthe2009diophantine}}]
  \label{prop_ostro}
  Given $(q_i)_{i\in \mathbb{N}}$ the denominators of the convergents of any
  irrational $0 < \alpha < 1$, every positive integer N can be uniquely written
  as
  \[
  N = 1 + \sum_{i=1}^{m}n_iq_{i-1}
  \]
  where $\left\{ 
    \begin{array}{l}
      0\le n_1 \le k_1 - 1, 0\le n_i \le k_i, \text{~for~} i\ge 2,\\
      n_i = 0 \text{~if~} n_{i+1}=k_{i+1}
    \end{array}\right.$ (``Markovian'' conditions).
\end{property}

This number system associated to the $(q_i)_{i\in \mathbb{N}}$ is named the
Ostrowski number system. To write an integer in this number system, we use a
classical decomposition algorithm (Algorithm \ref{algo_ostro}). The rank $m$ is
chosen such that $q_{m} > N$.

\begin{algorithm}
  \SetKwInOut{Input}{input}
  \SetKwInOut{Output}{output}
  \Input{$N\in\mathbb{N}$, $(q_i)_{i < m}$}
  \Output{$n_i$ such that $\displaystyle{N = 1 + \sum_{i=1}^{m}n_iq_{i-1}}$}

  $tmp \leftarrow N - 1$\;
  $i \leftarrow m$\;
  \While{$i \ge 1$}
  {
    $n_i \leftarrow \floor{tmp / q_{i-1}}$\;
    $tmp \leftarrow tmp - n_iq_{i-1}$\;
    $i \leftarrow i - 1$\;
  }

  \caption{Integer decomposition in Ostrowski number system.}  
  \label{algo_ostro}
\end{algorithm}

\begin{property}[{\cite[Proposition 2]{berthe2009diophantine}}]
  \label{prop_eta}
  Given $(\eta_i)_{i\in\mathbb{N}}$ the sequence of the signed partial remainders
  of any irrational $0 < \alpha < 1$, every real $\beta$, with $0 \le \beta < 1$
  can be uniquely written as
  \[
  \beta = \alpha + \sum_{i=1}^{+\infty} b_i\eta_{i-1}
  \]
  where $\left\{ 
    \begin{array}{l}
      0\le b_1 \le k_1 - 1, 0\le b_i \le k_i, \text{~for~} i\ge 2,\\
      b_i = 0 \text{~if~} b_{i+1}=k_{i+1}
    \end{array}\right.$ (``Markovian'' conditions).
\end{property}

There also exists two other number systems that are dual to these two. One
decomposes integers in the basis $((-1)^iq_i)_{i\in \mathbb{N}}$ and the other
decomposes reals in the basis of the unsigned partial remainders
$(\theta_i)_{i\in\mathbb{N}}$ \cite{berthe2009diophantine}. The second Markovian
condition then becomes $b_{i+1} = 0 \text{~if~} b_{i}=k_{i} $.
An algorithm to write real numbers in the $(\theta_i)_{i\in\mathbb{N}}$ number
scale has been proposed by Ito \cite{ito1986some}. It proceeds by iterating the
mapping $T_1 : (\alpha, \beta) \rightarrow (1/\alpha - \lfloor 1/\alpha \rfloor,
\beta/\alpha - \lfloor \beta/\alpha \rfloor)$.


\subsection{Related number systems over rational numbers}

In this subsection, we consider $\alpha = p/q$ rational. We recall that the
continued fraction expansion of a rational is finite. We denote
\[
\frac{p}{q} = [0; k_1, k_2, \dots, k_n]
\]
the continued fraction expansion of $p/q$, and recall $p_n = p$ and $q_n = q$.

The Ostrowski number system still holds for integers $N < q_n$, since the
keypoint in the Ostrowski number system is that there exists $q_m$ such that
$q_{m} > N$.

The $(\eta_i)_{i<n}$ number system also still holds under one supplemental
condition: $\beta$ must be rational with precision at most $q$ (i.e. the
denominator of $\beta$ must be less or equal than $q$).


\section{Modular arithmetic and continued fraction}
\label{sect_mod_arith}
In this section, we consider $\alpha = a/d$. We highlight that the same
decomposition $(b_1, \dots, b_{n+1})$ can be interpreted in two ways depending
on the number system used. In the Ostrowski number system, we obtain an integer
$N$ whereas in the number scale $(\eta_i)_{i\in \mathbb{N}}$, we obtain the
reduced value of $N\alpha \mod{1}$ \cite{berthe2009diophantine}. Hence, we will
use the fact that studying an integer $a$ modulo $d$ is similar to considering
the rational $a/d$ modulo $1$. This enables us to use properties
\ref{prop_ostro} and \ref{prop_eta} to compute modular multiplication and
division.

\subsection{Modular arithmetic and continued fraction}
\label{mod_arith_cf}

First, we briefly recall how continued fraction expansion and the Euclidean
algorithm are linked. We write $(\theta'_i)_{i\in\mathbb{N}}$ the integer
sequence of remainders when computing $\gcd(a,d)$. This sequence is composed of
decreasing values less than $d$. We also write $(\eta'_i)_{i\in\mathbb{N}}$ the
sequence $((-1)^i\theta'_i)_{i\in\mathbb{N}}$. We obtain the following
recurrence relation, and recall the recurrence relation over the
$(\theta_i)_{i\in\mathbb{N}}$ sequence of partial remainders of the continued
fraction expansion of $a/d$ :
$$
\begin{array}{lll}
  \theta'_{-1} = d & \theta'_0 = a & \theta'_i = \theta'_{i-2} -
  \lfloor \theta'_{i-2}/\theta'_{i-1}\rfloor \theta'_{i-1}\\
  \theta_{-1} = 1 & \theta_0 = a/d & \theta_i = \theta_{i-2} - \lfloor \theta_{i-2}/\theta_{i-1}\rfloor\theta_{i-1}.
\end{array}
$$

It is widely known and can be easily proved by induction that both sequences
compute the same partial quotients, that we will note $k_i$.

\begin{proof}[Proof of $k_{i+1} = \lfloor \theta_{i-1} / \theta_i \rfloor =
  \lfloor \theta'_{i-1} / \theta'_i\rfloor$]
  We prove it by proving $\theta_{i-1}/\theta_i = \theta'_{i-1} / \theta'_i$.
  \begin{itemize}
  \item {\bf Base case} : $\theta_{-1}/\theta_0 = d/a = \theta'_{-1}/\theta'_{0}$
  \item {\bf Induction} : Let $i$ such that $\theta_{i-1}/\theta_i =
    \theta'_{i-1}/\theta'_i$.
    \begin{align*}
      \frac {\theta_{i-1}}{\theta_i} & = \frac {\theta'_{i-1}}{\theta'_i} \\
      \frac{\theta_{i+1} + \lfloor \theta_{i-1}/\theta_i \rfloor
        \theta_i}{\theta_i} &= \frac{\theta'_{i+1} + \lfloor
        \theta'_{i-1}/\theta'_i \rfloor \theta'_i}{\theta'_i}\\
      \frac{\theta_{i+1}}{\theta_i} + \lfloor \theta_{i-1}/\theta_i \rfloor & = 
      \frac{\theta'_{i+1}}{\theta'_i} + \lfloor \theta'_{i-1}/\theta'_i \rfloor
    \end{align*}
    which implies ${\theta_i}/{\theta_{i+1}} =
    {\theta'_i}/{\theta'_{i+1}}$. \qedhere
  \end{itemize}
\end{proof}

It can also be noticed that $\eta'_i = \eta_id$. Actuallly, $\theta'_i =
\theta_id$ as the extended Euclidean algorithm compute the relations $\theta'_i
= (-1)^i(q_ia - p_id)$. In particular, it gives the Bezout's identity with
$\theta'_{n-1} = (-1)^{n-1}(q_{n-1}a - p_{n-1}d) = \gcd(a,d)$, and $q_{n-1}$
the inverse of $a$ if $a$ is invertible modulo $d$ ($\gcd(a,d) = 1$).

\subsection{Modular multiplication}

Now, given $a, b \in \mathbb{Z}/d\mathbb{Z}$, we write $c = a\cdot b \mod{d}$
the integer $0 \le c < d$ such that $ab - \lfloor ab/d \rfloor\cdot d = c$.

We can observe that the decompositions presented in properties \ref{prop_ostro}
and \ref{prop_eta} are both unique and both need the same ``Markovian''
condition over their coefficients. Hence, we can interpret the same
decomposition in both basis.

\begin{theorem}
  \label{theo_ch_basis}
  Given $a, b \in \mathbb{Z}/d\mathbb{Z}$, and $(q_i)_{i\le n}$, $(\eta'_i)_{i\le
    n}$ from Euclidean algorithm on $a$ and $d$, if we write $b$ in the
  $(q_i)_{i\le n}$ number scale as
  $$
  b = 1 + \sum_{i=1}^{n+1}b_iq_{i-1},
  $$ 
  then
  $$
  a\cdot b \mod{d} = a + \sum_{i=1}^{n+1}b_i\eta'_{i-1}.
  $$
\end{theorem}

\begin{proof}
  First, we consider $b < q_n$, it can be written in the Ostrowski number system
  as
  $$
  b = 1 + \sum_{i=1}^{n}b_iq_{i-1},
  $$
  and the coefficients $b_i$ respect the ``Markovian'' condition of the
  Ostrowski number system. Hence,
  $$
  \alpha\cdot b = \alpha + \sum_{i=1}^{n}b_iq_{i-1}\alpha.
  $$
  By definition, $\eta_i = q_i\alpha - p_i$, thus
  $$
  \alpha\cdot b = \alpha + \sum_{i=1}^{n}b_i\eta_{i-1} +
  \sum_{i=1}^{n}b_ip_{i-1}.
  $$
  As the coefficients $b_i$'s verify the ``Markovian'' condition, the uniqueness
  of the decomposition in property \ref{prop_eta} gives \mbox{$0 \le \alpha +
    \sum_{i=1}^{n}b_i\eta_{i-1} < 1$} and $\sum_{i=1}^{n}b_ip_{i-1} \in
  \mathbb{N}$.  Hence,
  $$
  \alpha\cdot b \mod{1} = \alpha + \sum_{i=1}^{n}b_i\eta_{i-1}.
  $$
  By multiplying this inequality by $d$, as $\alpha=a/d$ and $\eta'_i =
  \eta_id$, we obtain
  $$
  a\cdot b \mod{d} = a + \sum_{i=1}^{n}b_i\eta'_{i-1}.
  $$
  which finalizes the proof of the theorem for $b < q_n$.
  
  Now if $b\ge q_n$ and $b = b_{n+1}q_n + b'$ with $b' < q_n$ the
  remainder of the division of $b$ by $q_n$, $b'$ can be uniquely written in
  the Ostrowski number system. Furthermore, as $\eta'_n = 0$, $b_{n+1}\eta'_{n}
  = 0$, which finishes the proof.
\end{proof}

\subsection{Modular division}

Inversely, given $a, b \in \mathbb{Z}/d\mathbb{Z}$, with $a$ invertible modulo
$d$ ($\gcd(a,d) = 1$) we can efficiently compute \mbox{$a^{-1}\cdot b \mod{d}$}.

\begin{theorem}
  \label{theo_eta_to_qi}

  Given $a, b \in \mathbb{Z}/d\mathbb{Z}$ with $\gcd(a,d) = 1$, and $(q_i)_{i\le
    n}$, $(\theta'_i)_{i\le n}$ from Euclidean algorithm on $a$ and $d$, if we
  write $b$ in the $(\theta'_i)_{i<n}$ number scale as
  $$
  b = \sum_{i=1}^{n+1}b_i\theta'_{i-1},
  $$ 
  then if we denote $c = \displaystyle{\sum_{i=1}^{n+1}b_i(-1)^{i-1}q_{i-1}}$,
  $$  
  a^{-1} \cdot b \mod{d} \in \{c, d+c \}.
  $$
\end{theorem}

\begin{proof}
  The proof of correctness is similar to the one of theorem \ref{theo_ch_basis},
  using the facts that $\theta'_i = \theta_id$ and that $\theta_i =
  (-1)^i(q_i\alpha - p_i)$.

  Now, the greatest integer $c$ is clearly the one associated to the
  decomposition $(k_1, 0, k_3, 0, \dots, k_{n})$ when $n$ is odd. However,
  $k_iq_{i-1} = q_i - q_{i-2}$ by definition, which implies
  $$\sum_{i=0}^{(n-1)/2} k_{2i+1}q_{2i} = q_n.$$

  The smallest integer that can be returned is clearly the one associated to
  the decomposition $( 0, k_2, 0, k_4, \dots, k_{n})$ when $n$ is even.
  Once again, as $k_iq_{i-1} = q_i - q_{i-2}$, we get
  $$ - \sum_{i=1}^{n/2} k_{2i}q_{2i-1} = 1 - q_n.$$

  Hence, $-d < \sum_{i=1}^{n+1}b_i(-1)^{i-1}q_{i-1} < d$, that is to say, the
  result needs at most a correction by an addition by $d$.
\end{proof}

We mention that we also tried to decompose $b$ in the $(\eta'_i)_{i\le n}$
signed remainders number scale and evaluate this same decomposition in the
$(q_i)_{i \le n}$ number scale to compute modular division. We used Ito $T_2$
transform \cite{ito1986some} $T_2 : (\alpha, \beta) \rightarrow (1/\alpha -
\lfloor 1/\alpha \rfloor, \lceil \beta/\alpha \rceil - \beta/\alpha)$. In
practice, it returns the right result without the need of any
correction. However, as the decomposition computed by Ito $T_2$ transform does
not verify the same ``Markovian'' conditions as in the Ostrowski number system,
we were not able to give a theoretical proof that it always returns the reduced
result of the modular division.


\section{Elements of Complexity Analysis}
\label{sect_comp}
In this section, we introduce elements of complexity analysis of the proposed
modular multiplication algorithm based on theorem \ref{theo_ch_basis}. The same
analysis holds for the division.

First, the algorithm computes $(q_i)_{i\le n}$ and $(\eta'_i)_{i\le n}$.  This
can be computed using the classical extended Euclidean algorithm in
$O(\log{(d)}^2)$ binary operations. We notice here that the divisions computed
in the Euclidean algorithm can be computed by subtraction as the mean computed
quotient equals to Khinchin's constant (approximately $2.69$)
\cite[p. 93]{Khinchin92}. Furthermore, big quotients are very unlikely to occur
as the quotients of any continued fraction follow the Gauss-Kuzmin distribution
\cite[p. 83]{Khinchin92} \cite[p. 352]{TAOCP2},
$$
\mathbb{P}(k_i = k) = -\log_2\left( 1 - \frac{1}{(k+1)^2} \right).
$$

Second, the decomposition in $(q_i)_{i\le n}$ as in algorithm~\ref{algo_ostro}
also clearly has complexity in $O(\log{(d)}^2)$.  By the same arguments, the
coefficients of the decomposition in $(q_i)_{i\le n}$ can be computed by
subtraction as they are likely small.  The only quotient not following the
Gauss-Kuzmin distribution is the coefficient $b_{n+1}$ as it corresponds to the
quotient $\lfloor b/q_n \rfloor$. We prove in \ref{app_A} that if $a, d$ are
uniformly chosen integers in $[1,N]$ and $b$ is uniformly chosen in $[1,d]$,
then when $N$ tends to infinity, $\mathbb{P}(b_{n+1} \le k )$ tends to
\[
\zeta (2)^{-1} \left[ \sum_{i=1}^{k+1}\frac{i - (k+1)}{i^3} + (k+1)\zeta(3) \right].
\]

\begin{figure}
  \centering
  \begin{tikzpicture}
    \begin{axis}[
      width=0.7\linewidth,
      grid = major,
      xtick={0,5,...,64},
      xmin=0, xmax=49,
      minor x tick num={4},
      ytick={0.7, 0.75, ...,1},
      ymin=0.70, ymax=1,
      minor y tick num={4},
      xlabel=Max expected $b_{n+1}$, 
      ylabel=Probability,
      ]
      \addplot[mark=.] file {proba.dat};
    \end{axis}
  \end{tikzpicture}

  \caption{Probability law of the value of the coefficient $b_{n+1}$}
  \label{fig_proba}
\end{figure}
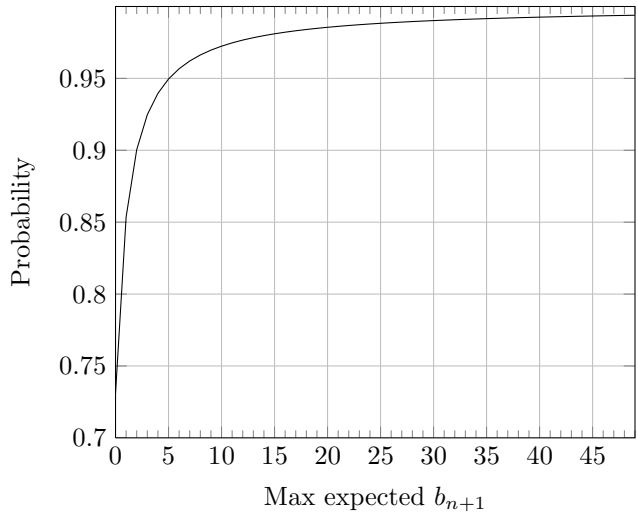

Figure \ref{fig_proba} shows the probability distribution of $\mathbb{P}(b_{n+1}
\le k)$. In particular, we obtain \mbox{$\mathbb{P}(b_{n+1} \le 3 ) \approx 92.5
  \%$}.

To finish the complexity analysis, evaluating the sum to return the final result
can also be done in $O(\log{(d)}^2)$.


\section{Perspectives}
\label{sect_persp}

In this paper, we presented an algorithm for modular multiplication and an
algorithm for modular division. Both are based on the extended Euclidean
algorithm and are of quadratic complexity in the size of the modulus.

Furthermore, the two stated theorems imply that, knowing the remainders
generated when computing the $\gcd$ of a number $a$ and the modulus $d$, one can
compute efficiently reduced multiplications by $a$ or $a^{-1}$. This can be
useful in algorithms computing several multiplications and/or divisions by the
same number $a$, as in the Gaussian elimination algorithm for example.

The presented algorithms can also be useful in hardware implementation of
modular arithmetic. They allow to perform inversion, multiplication and division
with the same circuit.

Further investigations have to be led to find optimal decomposition algorithms,
that minimize the number of coefficients of the produced decomposition and their
size. Also, we are working on an efficient software implementation of these
algorithms.


\section{Aknowledgement}
This work was supported by the TaMaDi project of the french ANR (grant ANR 2010
BLAN 0203 01). This work has also been greatly supported and improved by many
helpful proof readings and discussions with Jean-Claude Bajard, Valérie Berthé,
Pierre Fortin, Stef Graillat and Emmanuel Prouff.

\section*{References}
\bibliographystyle{elsarticle-num} 
\bibliography{MyCollection}

\appendix
\section{Detailed proof of the distribution function of \mbox{$\{b_{n+1} <
    k\}$.}}
\label{app_A}

Let $U_1, U_2$ and $U_3$ be three independent uniform distributions over
$[0,1]$.  We write $a = \lceil U_1N \rceil$, $d = \lceil U_2N \rceil$ and $b =
\lceil U_3d \rceil$. We denote $A = \{b < (k+1)q_n\}$, $B = \{ \gcd (a,d) \le k
+ 1\}$, $\bar{B} = \{\gcd (a,d) > k + 1\}$ and ${B_i} = \{ \gcd(a,d) = i
\}$. Hence using the law of total probability we have
\begin{align*}
  \mathbb{P}(A) &= \mathbb{P}(A \cap B) + \mathbb{P}(A \cap \bar{B}), \\
  &= \bigsqcup_{i \le k+1}\mathbb{P}(A \cap B_i) + \bigsqcup_{i >
    k+1}\mathbb{P}(A \cap {B_i}),\\
  &= \bigsqcup_{i \le k+1}\!\!\mathbb{P}(A |B_i)\!\cdot\!\mathbb{P}(B_i) +
  \bigsqcup_{i > k+1}\!\!\mathbb{P}(A | {B_i})\!\cdot\!\mathbb{P}({B_i}).
\end{align*}

As the ${B_i}$ are disjoint events, we have
\[
\mathbb{P}(A) = \sum_{i= 1}^{k+1}~\mathbb{P}(A | B_i)\cdot\mathbb{P}(B_i) +
\sum_{i= k+2}^{+\infty}~\mathbb{P}(A | {B_i})\cdot\mathbb{P}({B_i}).
\]

First, $\mathbb{P}(A | B_i) = 1$ for $i \le k+1$ as $b < d = \gcd(a,d)\cdot q_n
\le (k+1)\cdot q_n$. Hence,
\[
\mathbb{P}(A) = \sum_{i= 1}^{k+1}\mathbb{P}(B_i) + \sum_{i=
  k+2}^{+\infty}~\mathbb{P}(A | {B_i})\cdot\mathbb{P}({B_i}).
\]

Now we want to determine $\mathbb{P}(A | B_i)$ for $i \ge k+2$. Hereafter, we
write $\mathbb{Q}_i(\cdot) = \mathbb{P}(\cdot | B_i)$ and
\begin{align*}
  \mathbb{P}(A|B_i) &= \mathbb{Q}_i(A), \\
  &= \sum_{l=1}^N \sum_{m=1}^N \mathbb{Q}_i(\{a=l\} \cap \{d=m\})\cdot
  \mathbb{Q}_i(A~|~\{a=l\} \cap \{d=m\}).
\end{align*}
However, 
\[
\mathbb{Q}_i(A~|~\{a=l\} \cap \{d=m\}) = \frac{k+1}{i}
\]
as $b$ is uniformly distributed between $1$ and $d = iq_n$. If we consider the
segment of length $d$ and slice it in $i$ segments of length $q_n$, it can be
interpreted as the probability that $b$ is in the first $k+1$ slices. Hence
\begin{align*}
  \mathbb{P}(A|B_i) &= \sum_{l=1}^N \sum_{m=1}^N \mathbb{Q}_i(\{a=l\} \cap
  \{d=m\})\cdot \frac{k+1}{i},\\
  &= \frac{k+1}{i} \cdot \sum_{l=1}^N \sum_{m=1}^N \mathbb{Q}_i(\{a=l\} \cap
  \{d=m\}).
\end{align*}

As $\{a=l\}$ and $\{d=m\}$ are independent by hypothesis ($U_1$ and $U_2$ are
independent), $$\mathbb{Q}_i(\{a=l\} \cap
\{d=m\})=
\mathbb{Q}_i(\{a=l\})\cdot\mathbb{Q}_i(\{d=m\}),$$ and
$$
  \mathbb{P}(A|B_i) = \frac{k+1}{i} \cdot \sum_{l=1}^N \mathbb{Q}_i(\{a=l\})
  \cdot\!\sum_{m=1}^N \mathbb{Q}_i(\{d=m\}).
$$

Now, we use the fact that the sum of the probabilities over the whole sample
space always sum to $1$ to obtain
$$
  \mathbb{P}(A|B_i) = \frac{k+1}{i}.
$$

If we recapitulate, 
\[
\mathbb{P}(A) = \sum_{i= 1}^{k+1}~\mathbb{P}(B_i) + \sum_{i=
  k+2}^{+\infty}~\frac{k+1}{i}\cdot\mathbb{P}({B_i}).
\]

Finally, it is widely known that $\mathbb{P}({B_i})$ tends to
$\frac{\zeta(2)^{-1}}{i^2}$ when $N$ tends to infinity
\cite[p. 353]{ITN_HW}. Hence, we get
\begin{align*}
  \lim_{N\rightarrow +\infty}\mathbb{P}(A) &=
  \sum_{i=1}^{k+1}\frac{\zeta(2)^{-1}}{i^2} + \sum_{i=k+2}^{\infty}
  \frac{k+1}{i}\cdot\frac{\zeta(2)^{-1}}{i^2},\\
  &= \zeta(2)^{-1}\left[ \sum_{i=1}^{k+1}\frac{1}{i^2} +
    (k+1)\sum_{i=k+2}^{+\infty}\frac{1}{i^3} \right],
\end{align*}
which equals to
\begin{align*}
  & \zeta(2)^{-1}\left[ \sum_{i=1}^{k+1}\frac{1}{i^2} + (k+1)\left(
      \sum_{i=1}^{+\infty}\frac{1}{i^3} - \sum_{i=1}^{k+1}\frac{1}{i^3} \right)
  \right],\\
  =\;& \zeta(2)^{-1}\left[ \sum_{i=1}^{k+1}\frac{i - (k + 1)}{i^3} + (k+1)\left(
      \sum_{i=1}^{+\infty}\frac{1}{i^3} \right) \right].
\end{align*}

By definition, Riemann zeta function equals
\[
\zeta(s) = \sum_{i=1}^{+\infty}\frac{1}{i^s}.
\]
Hence we get the following simplification, which is more convenient for
computation and has been used to generate Fig. \ref{fig_proba},
$$
\lim_{N\rightarrow +\infty}\!\!\!\mathbb{P}(A) \!=\! \zeta(2)^{-1}\!\left[
  \sum_{i=1}^{k+1}\!\frac{i - (k +
    1)}{i^3}\!+\!(k\!+\!1)\!\cdot\!\zeta(3)\right]\!.
$$


\listoftodos

\end{document}